\documentclass[letterpaper,12pt]{article}
\usepackage[lmargin=3cm,rmargin=3cm,bmargin=3cm,tmargin=2.5cm]{geometry}

\usepackage{fancyhdr}
\usepackage{bm,mathrsfs}
\usepackage{graphicx}
\usepackage{natbib}%[numbers]
\usepackage{amsmath,amssymb,amsgen,amsbsy,amsthm}
\usepackage{color,makecell}
\usepackage{tabularx}
\usepackage[figuresright]{rotating}

\newcommand{\Rey}{\mathrm{Re}}
\newcommand{\Gr}{\mathrm{Gr}}
\newcommand{\Rangle}{\right\rangle}
\newcommand{\Langle}{\left\langle}
\newcommand{\Sum}[1]{\sum_{#1=1}^\infty}
\newcommand{\phexp}{{\vphantom{*}}} %phantom exponent
\newcommand{\bel}{\begin{equation}\label}
\newcommand{\ee}{\end{equation}}
\newcommand{\beq}{\begin{eqnarray}\label}
\newcommand{\eeq}{\end{eqnarray}}

\newtheorem{theorem}{Theorem}%[section]
\newtheorem{lemma}{Lemma}%[section]
\newtheorem{remark}{Remark}%[section]

\allowdisplaybreaks

\begin{document}

\title{\bfseries\large\color{blue} How to extract a spectrum from hydrodynamic equations}

\author{John D. Gibbon$^1$\thanks{E-mail: \texttt{j.d.gibbon@ic.ac.uk}}
and Dario Vincenzi$^2$\thanks{E-mail: \texttt{dario.vincenzi@univ-cotedazur.fr}.
Also Associate, International Centre for Theoretical Sciences,
Tata Institute of Fundamental Research, Bangalore 560089, India}
\\
$^1${\it\normalsize Department of Mathematics, Imperial College London SW7 2AZ, UK}
\\
$^2${\it\normalsize Universit\'e C\^ote d'Azur, CNRS, LJAD, 06100 Nice, France}
}
\maketitle
\vspace{-10mm}
\begin{center}
{\em\color{blue}A paper dedicated to the memory of Charles Doering (1956--2021)}
\end{center}

\begin{abstract}
Practical results gained from statistical theories of turbulence usually appear in the form of an inertial range energy spectrum $\mathcal{E}(k)\sim k^{-q}$ and a cut-off wave-number $k_{c}$. For example, the values $q=5/3$ and $\ell k_{c}\sim \Rey^{3/4}$ are intimately associated with Kolmogorov's 1941 theory. To extract such spectral information from the Navier--Stokes equations, \cite{dg02} introduced the idea of forming a set of dynamic wave-numbers $\kappa_n(t)$ from ratios of norms of solutions. The time averages of the $\kappa_n(t)$ can be interpreted as the 2$n$th-moments of the energy spectrum. They found that $1 < q \leqslant 8/3$, thereby confirming the earlier work of \cite{sf75} who showed that when spatial intermittency is included, no inertial range can exist in the limit of vanishing viscosity unless $q \leqslant 8/3$. Since the $\kappa_n(t)$ are based on Navier--Stokes weak solutions, this approach connects empirical predictions of the energy spectrum with the mathematical analysis of the Navier--Stokes equations. This method is developed to show how it can be applied to many hydrodynamic models such as the two dimensional Navier--Stokes equations (in both the direct- and inverse-cascade regimes), the forced Burgers equation and shell models.
\end{abstract}
 
%%%%%%%%%%%%%%%%%%%%%%%
%\pacs{...}

\vspace{-3mm}
\section{\large\color{blue}Introduction}\label{sect:introduction}

The energy spectrum of the velocity field plays an important role in fluid dynamics, since it describes how kinetic energy distributes across scales.
In turbulent flows, the energy spectrum generally behaves as a power law in the range between the forcing and dissipation characteristic wavenumbers, with a slope that depends critically on the space dimension.  In view of their highly fluctuating nature, turbulent flows have been studied with statistical tools, and the form of the energy spectrum has been predicted by using dimensional analysis, renormalization-group techniques, and stochastic or closure models. For a recent review of this topic, the reader is referred to \citet{ab18}
and \citet{verma2019}.
%\citealp{ab18} and \citealp{verma2019}).
\par\smallskip
Establishing a rigorous connection between the statistical theory of turbulence and the mathematical analysis of the Navier--Stokes equations
is a difficult problem [\citealp{dg95}; \citealp{foias2001}; \citealp{constantin2006}; \citealp{doering2009}; Kuksin and Shirikyan,
\citeyear{ks12}; \citealp{bt13}]. Let us first summarize how empirical estimates for length scales in the statistical theory of homogeneous and isotropic turbulence have traditionally been obtained in terms of the energy spectrum. In a $d$-dimensional space, this is defined as
\bel{Edef}
\mathcal{E}(k)= c_d\, k^{d-1}\operatorname{Tr}\mathbb{F}(k), 
\ee
where $c_d$ is a positive constant which depends on the spatial dimension and 
\bel{Fdef}
\mathbb{F}_{ij}(k)=\int_{\mathbb{R}^d} e^{-i\bm k\cdot\bm r} \; \overline{\bm u(\bm x+\bm r,t)\cdot\bm u(\bm x,t)}\; dV_{r}
\ee
is the Fourier transform of the velocity spatial correlation function \citep{my75}. The overline denotes an ensemble average over the realizations of the velocity field in the statistically steady state. For a statistically stationary, homogeneous, and isotropic field, the spatial correlation does not depend on time and the position $\bm x$, but only on the separation $\bm r$. For $d=3$ assume that $\mathcal{E}(k)$ has an inertial range between the forcing wavenumber $\ell^{-1}$ and a cut-off wavenumber $k_{c}$ 
of the form
\bel{Eqdef}
\mathcal{E}(k)\sim \epsilon^{2/3}\ell^{q-5/3} k^{-q}\qquad (1<q<3)\,,
\ee
where 
\begin{equation}
\epsilon = \nu \int_{0}^{\infty} k^2 \mathcal{E}(k)dk %\sim \nu A k_{c}^{3-q}
\end{equation}
is the mean energy dissipation rate. By using \eqref{Eqdef} and ignoring the energy content in the range $k>k_c$, the mean energy dissipation rate can be estimated as $\epsilon^{1/3}\sim\nu\ell^{5/3-q}k_c^{3-q}$. This, together with the empirical prediction $\epsilon\sim U^3/\ell$  yields
\bel{kcdef}
\ell k_{c}\sim \Rey^{\frac{1}{3-q}}\,,
\ee
which can be found in \citet{frisch1995}. Here $U$ is the root-mean square velocity and $\Rey=U\ell/\nu$ is the Reynolds number.
The $2n$-th moment of the energy spectrum, i.e.
\bel{eq:K}
K_n^{2n}=\dfrac{\int_0^\infty k^{2n} \mathcal{E}(k)\,dk}{\int_0^\infty \mathcal{E}(k)\,dk} \qquad (n\geqslant 1)\,,
\ee
is then estimated as
\bel{Knqest}
\ell K_n \sim (\ell k_{c})^{1-\frac{q-1}{2n}} \sim \Rey^{\frac{1}{3-q}-\frac{1}{2n}\left(\frac{q-1}{3-q}\right)}\,.
\ee
Kolmogorov's 1941 theory sets $q$ to 5/3, which gives 
\bel{kcest}
\ell k_{c}\sim\Rey^{3/4}\qquad\mbox{and}\qquad\ell K_n\sim\Rey^{3/4-1/4n}\,.
\ee
\par\vspace{0mm}
%%%%%%% Start of Table %%%%%%
{\small
%\begin{sidewaystable}
\begin{table}
\caption{\scriptsize Estimates for the time average of $\mathcal{L}\langle\kappa_n\rangle$ and corresponding predictions for the inertial-range energy spectrum. $\mathcal{L}$ is the box size $L$ for $d=2$ and the forcing length scale ($\ell$) in all the other cases. Unless otherwise specified, $q > 1$.}\label{table}
\medskip\centering
\renewcommand{\arraystretch}{2}
\begin{tabular}{|c|c|c|}                                                                                                 
\hline
System & Upper bounds on $\mathcal{L}\langle\kappa_n\rangle$ & $\mathcal{E}(k)\sim k^{-q}$\\\hline\hline
3D Navier--Stokes & $a_{\ell}^{3-\frac{7}{2n}}\,\Rey^{3-\frac{5}{2n}+\frac{\delta}{n}}$ & $q\leqslant \frac{8}{3}$
\\
\hline
\makecell{3D Navier--Stokes \\ with suppressed fluctuations}
& $a_{\ell}^{\frac{3(n-1)(p-2)}{n(p+6)}-\frac{1}{2n}}\, \Rey^{\frac{6np-5p+6}{2n(p+6)}}$ & $q\leqslant 
\frac{8}{3} - \frac{2}{p}$
\\
\hline
2D Navier--Stokes (direct cascade)  & $a_{\ell}^{\frac{3}{2}\left(1-\frac{1}{n}\right)}\, \Rey^{\frac{3}{4}-\frac{1}{2n}}$ & 
$q\leqslant \frac{11}{3}$   
\\
\hline
\makecell{2D Navier--Stokes (direct cascade) with\\
monochromatic or constant $\epsilon$-forcing}
& $a_{\ell}^{\frac{3}{2}\left(1-\frac{1}{n}\right)}\, \Rey^{\frac{1}{2}}$ & $q\leqslant 3$   
\\
\hline
2D Navier--Stokes (inverse cascade) & $a_\ell^{n/2}\, \Rey^{\frac{1}{2}}$ & $\frac{5}{3}\leqslant q$   
\\
\hline
Burgers & $a_{\ell}^{\frac{1}{3}-\frac{5}{6n}}\,\Rey^{1-\frac{1}{2n}}$ & $q\leqslant 2$
\\
\hline
Shell model & $a_\ell^{-\frac{1}{2n}}\,\Rey^{\frac{3}{4}-\frac{1}{4n}}$ & $q\leqslant \frac{5}{3}$ 
\\
\hline
\end{tabular}
%\end{sidewaystable}
\end{table}%%%End of the table 
}
\par\smallskip\noindent
How can a result like (\ref{Knqest}) be achieved for the incompressible Navier--Stokes equations? More specifically, how can the value of $q$ be determined from the analysis? Rigorous results for partial differential equations are conventionally expressed as estimates of time-averages of spatial norms and not in terms of spectra. Indeed, in the language of Sobolev norms the idea of a spectrum associated with an inertial range, as in (\ref{Eqdef}), has no meaning. How to circumvent this difficulty and extract results corresponding to (\ref{Knqest}) for weak solutions of the three-dimensional Navier--Stokes equations was first addressed by \cite{dg02} twenty years ago. Moreover, in a separate but parallel paper, \citet{df02} also addressed how length scales in the forcing can be used to achieve estimates in terms of the Reynolds number $\Rey$ instead of the less physical Grashof number $\Gr$\,: see (\ref{GRdef}) in \S\ref{sect:NS} for definitions of these dimensionless quantities. 
\par\smallskip
A summary of these ideas is the following\,: first write down the Navier--Stokes equations on a periodic $d$-dimensional domain $\mathcal{V}=[0,L]^d$, where $d=2,\,3$
\bel{NSE1}
\partial_{t}\bm u +\bm u\cdot\nabla\bm u = -\nabla p+\nu\Delta\bm u + \bm f(\bm x)\,, \qquad \nabla\cdot\bm u = 0\,.
\ee
Here $\bm u(\bm x,t)$ is the velocity field, $p$ is pressure, $\nu$ is the kinematic viscosity, and $\bm f(\bm x)$ is a time-independent, mean-zero, and divergence-free body forcing. For simplicity, we follow \citet{df02} in assuming that the forcing is narrow-band, i.e.~it is concentrated on wavenumbers $k\sim\ell^{-1}$. Parseval's equality then implies that $\Vert\bm f\Vert_2\approx \ell^{n}\Vert\nabla^n\bm f\Vert_2$, where $\Vert\cdot\Vert_2^2=\int_{\mathcal{V}} \vert\cdot\vert^2\, dV$. The aspect ratio of the box size to the forcing scale is denoted as
\bel{aelldef}
a_{\ell} = \frac{L}{\ell}.
\ee
As a consequence of Poincar\'e's inequality, $a_\ell\geqslant 2\pi$. The initial velocity field is taken mean-zero, so that $\bm u(\bm x,t)$ remains mean-zero at all times. In \citet{dg02} the following sequence of squared $L^2$-norms was introduced 
\bel{Fndef}
F_n(t)= H_n(t) +\tau_{n}^{2}\,\Vert\nabla^n\bm f\Vert_2^2, \qquad n=1,2,\dots
\ee
with 
\bel{Hndef}
H_n(t)=\Vert\nabla^n\bm u(\cdot,t)\Vert_2^2\,.
\ee
The forcing term is included for the technical reason that the analysis involves division by $F_n$ and thus $H_n$ may be small on certain time intervals. The time scales\footnote{In \citet{dg02} the $\tau_{n}$ had no $n$-dependence. However, the development of the method to other cases sometimes requires this dependence so it has been introduced at this point.} $\tau_n$ are chosen in such a way that the contribution of the forcing does not dominate the time average of $H_n$ in the turbulent regime so the $\Rey$-scaling of the time averages of $F_n$ and $H_n$ remains the same. These technical issues are addressed in \S\ref{sect:NS}. Then the following family of time-dependent ratios was introduced
\bel{eq:ratios}
\kappa_{n,r}(t)=\left(\frac{F_n}{F_r}\right)^\frac{1}{2(n-r)} \qquad (0\leqslant r<n)\,.
\ee
The $\kappa_{n,r}$ have the dimension of a wavenumber and are ordered according to $\kappa_{n,r}\leqslant\kappa_{n+1,r}$ and $\kappa_{n,r}\leqslant\kappa_{n,r+1}$. The quantities $\kappa_n\equiv\kappa_{n,0}$ play a special role because of their physical meaning. Indeed, Parseval's equality yields
\bel{Hndef-Fourier}
H_n(t) = L^{-d}\sum_{\bm k} k^{2n}\vert \hat{\bm u}(\bm k,\,t)\vert^{2}
\ee
with $\hat{\bm u}(\bm k,t)$ as the inverse spatial Fourier transform of ${\bm u}(\bm x,\,t)$. Hence
\bel{kndef}
\kappa_{n}^{2n}(t) =\dfrac{\sum_{\bm k} k^{2n}\left(\vert \hat{\bm u}(\bm k,t)\vert^{2}+\tau_{n}^{2}\,\vert \hat{\bm f}(\bm k)\vert^{2}\right)}{\sum_{\bm k}\left(\vert \hat{\bm u}(\bm k,t)\vert^{2}+\tau_{n}^{2}\,\vert \hat{\bm f}(\bm k)\vert^{2}\right)}\,.
%\qquad \text{as $\Gr\to\infty$}.
\ee
At large Reynolds numbers, $\kappa_{n}^{2n}(t)$ can therefore be regarded as the $2n$-th moment of the (instantaneous) energy spectrum.
The strategy in \citet{dg02}, which also adopted some ideas on the forcing from \citet{df02}, was to find a class of estimates of the type 
\bel{ta1} 
\langle\ell\kappa_{n}\rangle \leqslant c_{n}\Rey^{\xi_{n}}\,,
\ee
for the set of time averages $\langle\ell\kappa_{n}\rangle$, where the brackets 
\bel{tadef}
\Langle\cdot\Rangle=\limsup_{T\to\infty}\frac{1}{T}\int_0^T \cdot\; \mathrm{d}t
\ee
denote a long-time average.  The specific form of $\xi_{n}$ found in \citet{dg02} is given in Theorem \ref{th:dg02} and is also displayed in the first line of Table \ref{table}.  The estimate in (\ref{ta1}) in terms of $\Rey$ then allowed them to make the final step which was to compare the exponent $\xi_{n}$ with that in (\ref{Knqest}) that comes from \citet{frisch1995}
\bel{exp1}
\frac{1}{3-q}-\frac{1}{2n}\left(\frac{q-1}{3-q}\right) \leqslant \xi_{n}\,.
\ee
The direction of the inequality in (\ref{exp1}) reflects that in (\ref{ta1}). In reality, results from statistically stationery, homogeneous, isotropic turbulence theory are being compared with estimates of the long-time averages of ratios of Navier-Stokes spatial norms. The value of $\xi_{n}$ from Theorem \ref{th:dg02} gives the range of $q$ and this turns out to be precisely 
\bel{qrange}
1 < q \leqslant 8/3
\ee
as in \citet{sf75}, where a bound on the energy spectrum was obtained by considering a `shell decomposition' of the velocity field and examining the energy flux across wavenumbers. 
\par\smallskip%\noindent
Here we will endeavour to show that this method has much greater scope and can be applied in other circumstances, such as the $2D$ Navier--Stokes equations (in both the direct- and inverse-cascade regimes), Burgers equation, and shell models. Table \ref{table} summarises the range of $q$ for each of these cases with the details provided in the rest of the paper. Although the spectral slopes for these systems are known, our study shows that they can be obtained in a systematic way within the same mathematical framework and thus confirms the wide applicability of these methods to the analysis of hydrodynamic equations. 

%%%%%%%%%%%%%

%%%%%%%%%%%%%%%%%%%%

\section{\large\color{blue}The Navier--Stokes equations in three and two dimensions}\label{sect:NS}

In the following, we consider the Navier--Stokes equations in both $d=3$ and $d=2$ dimensions. For weak solutions with initial data in $L^2(\mathcal{V})$, the root-mean square velocity 
\bel{Udef}
U=L^{-d/2}\sqrt{\Langle\Vert\bm u\Vert_2^2\Rangle} < \infty\,.
\ee
Likewise, the root-mean square of the (time-independent) forcing is $f=L^{-d/2}\Vert\bm f\Vert_2$. Suitable definitions of the Grashof and Reynolds numbers are 
\bel{GRdef}
\Gr=\frac{f\ell^3}{\nu^2}\,,\qquad\quad
\Rey=\frac{U\ell}{\nu}\,.
\ee
The former is a dimensionless measure of the magnitude of the forcing, whereas the latter is the system response. $\Gr$ and $\Rey$ satisfy the bound 
\citep{df02}
\bel{eq:Gr-Re}
\Gr\leqslant c(\Rey+\Rey^2)\,,
\ee
which shows that the turbulent regime is achieved for $\Gr\gg 1$. The bound in \eqref{eq:Gr-Re} can be rewritten in terms of the mean energy dissipation rate
\bel{epsilondef}
\epsilon=\nu L^{-d}\Langle H_1\Rangle
\ee
as
\bel{epbd1}
\epsilon \leqslant c\,\nu^3\ell^{-4}\left(\Rey^2+\Rey^3\right)\,.
\ee
\citet{df02} also proved the inequality
\bel{lower-eps}
c_1 f \leqslant c_2 \nu^{1/2}\ell^{-1}\epsilon^{1/2}+c_3 \nu^{-1/2} U \epsilon^{1/2}\,,
\ee
which, in turn, gives the lower bound
\bel{lower-eps-2}
\epsilon \geqslant c\,\nu^3\ell^{-4}\,\frac{\Gr^2}{(1+\Rey)^2}\,.
\ee
The $L^2$-norms $F_n$ include a contribution from the forcing which must not dominate $\langle H_n\rangle$ as $\Gr\to\infty$. This is achieved by suitably choosing the time scales $\tau_n$. Using Poincar\'e's inequality, the assumption of a narrow-band forcing, and \eqref{lower-eps-2} yields
\begin{align}
\label{eq:ratio}
\frac{\tau_n^2 \Vert\nabla^n\bm f\Vert_2^2}{\langle H_n\rangle}
&\leqslant c_n \, L^{2(n-1)} \frac{\tau_n^2 \Vert\nabla^n\bm f\Vert_2^2}{\langle H_1\rangle} 
= c_n\, \nu \epsilon^{-1} L^{2(n-1)}\ell^{-2n}\tau_n^2 f^2
\\
&= c_n\, \nu^5 \epsilon^{-1}L^{2(n-1)}\ell^{-2(n+3)}\tau_n^2\Gr^2
\leqslant c_n\, \nu^2 \ell^{-4}a_\ell^{2(n-1)}\tau_n^2(1+\Rey)^2\,.
\end{align}
Therefore, a suitable definition of $\tau_n$ is
\bel{eq:tau1}
\tau_n = \frac{\nu^{-1}\ell^{2}}{a_\ell^{(n-1)}(1+\Rey)^{1+2\delta}}
\ee
with 
\bel{deltadef1}
0 < \delta < \frac{1}{6}\quad\mbox{for}\quad d=3\qquad\mbox{and}\qquad \delta=0\quad\mbox{for}\quad d=2\,,
\ee
so that 
\bel{deltadef2}
\tau_{n}^{2} \Vert\nabla^n\bm f\Vert_2^2 \leqslant c_n \Rey^{-4\delta}\,\langle H_n\rangle
\qquad\mbox{as}\qquad\qquad\Gr\to\infty\,.
\ee
The non-zero $\delta$-correction is required when $d=3$ because, for technical reasons, the forcing contribution to $\langle F_n\rangle$ needs to become negligible as $\Gr\to\infty$ \citep{dg02}. When $d=2$, the contribution of the forcing simply must not grow faster than $\langle H_n\rangle$ \citep{gp07}. We shall see that the use of definition \eqref{eq:tau1} systematically improves the power of $a_\ell$ in the estimates of $\langle\kappa_n\rangle$. Although this is of little importance in most cases because generally $a_\ell=O(1)$, it becomes essential in the study of the inverse-cascade regime of the $2D$ Navier--Stokes equations, which is characterized by large values of $a_\ell$.
%{\color{red}The factor of a non-zero $\delta$ is a technical necessity. It can either be chosen %very small or, alternatively, $\tau$ could be %chosen in the form
%\bel{eq:tau2}
%%\tau=\ell^2\nu^{-1}\Gr^{-1/2-\delta} 
%\tau=\ell^2\nu^{-1}(\Gr\ln\Gr)^{-1/2}\,.
%%\tau=\ell^2\nu^{-1}\Gr^{1/2}(1+\ln\Gr)^{-1/2}
%\ee}
%
%In particular, with the above definition of $\tau$, as $\Gr\to\infty$
%\begin{equation}
%\langle F_1\rangle = \nu^{-1}L^d\epsilon + \tau^2 L^d \ell^{-2}f^2 \sim \nu^{-1}L^d\epsilon\,,
%\end{equation}
%whence
%\bel{eq:F1}
%\langle F_1\rangle \leqslant c\, \nu^2L^d\ell^{-4}\Rey^3\,.
%\ee
\par\smallskip%\noindent
Finally, when $d=2$ and $d=3$ the $F_n$ satisfy the following `ladder' of differential inequalities as $\Gr\to\infty$ (see qualifications in Appendix~\ref{appendix-A})\,:
\bel{eq:ladder}
\frac{1}{2}\dot{F}_n\leqslant -\nu F_{n+1}+c_n\left(\Vert\nabla\bm u\Vert_\infty\,+\tau_n^{-1}\right)F_n\,.
\ee
Now we shall see that estimates for $\langle\kappa_{n}\rangle$ differ, leading to different ranges of $q$. Throughout this paper $c$ and $c_{n}$ denote dimensionless, generic constants.

%%%%%%%%%%%%%%%

\subsection{\small\color{blue}Three examples involving the $3D$ Navier--Stokes equations}\label{3dNSE}

The main result of \citet{dg02} is an estimate for the time average of $\kappa_n$ for weak solutions of the $3D$ Navier-Stokes equations.%
\footnote{The difference in the power of $a_\ell$ compared to the original version of the theorem is due to the use of definition \eqref{eq:tau1}, which gives an extra factor $a_\ell^{-1}$ in the estimate of $\langle\kappa_1^2\rangle$\,: $\ell^2\langle\kappa_1^2\rangle\leqslant c\, a_\ell^{-1} \Rey^{1+2\delta}$ as $\Gr\to\infty$. The rest of the proof is unchanged.}
%The estimate expressed below for $\ell\langle\kappa_n\rangle$ has been made in the case when the very small $\delta$-term in the expression %for $\tau$ given in (\ref{eq:tau1}) has been neglected. 
%%%%%%%%%%%%%
\begin{theorem}[\cite{dg02}]\label{th:dg02}
For $n\geqslant 2$ and $0 < \delta < \frac{1}{6}$,
\bel{thm1}
\ell\langle\kappa_n\rangle\leqslant c_{n}\,a_{\ell}^{3-\frac{7}{2n}}\Rey^{3-\frac{5}{2n} +\frac{\delta}{n}}
\qquad\mbox{as}\qquad\Gr\to\infty\,.
\ee
\end{theorem}
%%%%%%%%%%%%%
\begin{remark}
Comparing the exponents of $\Rey$ in (\ref{thm1}) and (\ref{Knqest}) gives
\bel{match1}
\frac{1-\frac{q-1}{2n}}{3-q} \leqslant 3-\frac{5}{2n} +\frac{\delta}{n}\,,
\ee
whence 
\bel{match2}
q \leqslant \frac{8}{3} + \frac{2\delta}{3(n-1)}\,.
\ee
Thus, for every value of $n$, we find that $1 < q \leqslant \frac{8}{3}$ as advertised in Table \ref{table}. 
\end{remark}
%\begin{remark}
%The result of Theorem \ref{th:dg02} has been given without the $\delta$-correction in (\ref{eq:tau1}). If this is included and the form of $\tau$ %given in (\ref{eq:tau2}) is used then (\ref{thm1}) has a multiplicative extra term of $(\ln\Rey)^{\frac{1}{4n}}$. This can be taken into account %by writing $\ln\Rey < 1+ \ln\Rey$ and then using the inequality $\ln(1+\ln\Rey)< \ln\Rey$, This adds a $1/4n$-term to the right hand side of %(\ref{match}) from which we find that the range narrows slightly to $1< q \leqslant 29/11$. 
%It is possible that the $\delta$-correction is an artefect of the method.
%\end{remark}
\par\smallskip\noindent
\citet{dg02} also investigated how the energy spectrum is modified when the spatial fluctuations of the velocity gradients are suppressed
through the assumption
\bel{eq:m}
\Vert\nabla\bm u\Vert_\infty \approx c\, L^{-3/p}\Vert\nabla\bm u\Vert_{p}\,, \qquad 2\leqslant p\leqslant \infty\,.
\ee
For $p=2$, this means that as $\Rey$ increases, the maximum velocity scales as the root-mean square velocity. Higher values of $p$ correspond to a milder suppression of fluctuations, and \eqref{thm1} is recovered for $p=\infty$. With approximation \eqref{eq:m}, Theorem~\ref{th:dg02} is modified as follows\footnote{The exponent $\delta$ can be set to zero when assumption~\eqref{eq:m} is used.}
\begin{theorem}[\cite{dg02}]%[\citet{dg02}]
%%%%%%%%%%%%
\label{th:dg02-m}
Under assumption \eqref{eq:m} and for $n\geqslant 2$
\bel{thm1eqn1}
\ell\langle\kappa_n\rangle\leqslant c_{n}\,a_{\ell}^{\frac{3(n-1)(p-2)}{n(p+6)}-\frac{1}{2n}}
\Rey^{\frac{6np-5p+6}{2n(p+6)}}\qquad\mbox{as}\qquad \Gr\to\infty\,.
\ee
\end{theorem}
%%%%%%%%%%%%
\par\smallskip\noindent
If the energy spectrum is as above, an argument analogous to that used for $p=\infty$ shows that the scaling in Theorem~\ref{th:dg02-m} is consistent with 
\bel{qp1}
1<q\leqslant \frac{8}{3}-\frac{2}{p}\,.
\ee
In particular, $p=2$ yields the Kolmogorov spectrum $q=5/3$. More generally, by altering the value of $p$ in the range $2 \leqslant p \leqslant \infty$ we find that the upper bound of $q$, designated as $q_{ub}$, lies in the range 
\bel{qubrange}
\frac{5}{3} \leqslant q_{ub} \leqslant \frac{8}{3}\,.
\ee
These methods have also been applied to magnetohydrodynamic turbulence to show that the Iroshnikov--Kraichnan total-energy spectrum can be excluded when there is no cross-correlation between the velocity and magnetic fields \citep{ggkppppss16}.

%%%%%%%%%%%%%%%%%%%%%%%%
\subsection{\small\color{blue} The $2D$ Navier--Stokes equations}\label{2dNSE}

Consider the Navier--Stokes equations on the periodic square $\mathcal{V}=[0,L]^2$. The definitions introduced in \S\ref{sect:introduction}
extend unchanged to two dimensions ($d=2$). However, the absence of vortex stretching leads to a different estimate for the time average of $\kappa_n$. 
\par\smallskip%\noindent
Two-dimensional turbulence is characterized by a dual cascade consisting of a direct cascade of enstrophy 
(defined as $\Vert\bm\omega\Vert_2^2$ with $\bm\omega=\nabla\times\bm u$) from
$\ell^{-1}$ to high wavenumbers and an inverse cascade of energy from $\ell^{-1}$ to low wavenumbers 
[Kraichnan and Montgomery, \citeyear{km80}; \citealp{kg02}; \citealp{t02}; \citealp{be12}].
The enstrophy cascade ends at a cutoff wavenumber $k_c$, beyond which enstrophy is dissipated by viscosity.
In an unbounded domain or in a bounded domain before statistical equilibrium is established,
the energy cascade continues to extend to ever smaller  wavenumbers, and a quasi-steady spectrum forms at wavenumbers between the
inverse integral scale and $\ell^{-1}$.

We study the spectra of the two cascades separately by considering first the case $\ell\sim L/2\pi$ (direct cascade)
and then $\ell\ll L$ (inverse cascade). 

%The time scale that enters the definition of $\kappa_{n}$ should now be taken as follows \citep{gp07}
%\begin{equation}
%\tau^{-1}=\ell^{-2}\nu\mathit{Gr}^{1/2}(1+\ln\Gr)^{1/2}\leqslant c\,\ell^{-2}\nu \Rey(1+\ln\Rey)^{1/2}.
%\end{equation}
%%%%%%%%%%%%
\subsubsection{\small\color{blue}Direct cascade of enstrophy}\label{2Ddirect}

The following theorem%
\footnote{There is a small difference in the powers of $a_\ell$ and $\ln\Rey$ between \eqref{eq:kappa2-2d} and the original version of the theorem. This is due to the choice of $\tau_n$, which modifies the estimates of $\langle\kappa_1^2\rangle$ and $\langle\kappa_{2,1}^2\rangle$.
With definition \eqref{eq:tau1}, as $\Gr\to\infty$
$$
L^2\langle\kappa_1^2\rangle\leqslant c\, a_\ell \Rey\qquad \text{and} \qquad
L^2\langle\kappa_{2,1}^2\rangle\leqslant c\, a_\ell^2 \Rey \,.
$$}
describes the behaviour of $\langle\kappa_{n}^2\rangle$ as $\Gr\to\infty$ while $a_\ell=O(1)$.
\begin{theorem}[\cite{gp07}]
\label{th:gp07}
For $n\geqslant 2$
\bel{eq:kappa2-2d}
L^2\langle\kappa_n^2\rangle\leqslant c_n
a_{\ell}^{3\left(1-\frac{1}{n}\right)} \Rey^{\frac{3}{2}-\frac{1}{n}}\,[\ln(a_\ell^2\Rey)]^{\frac{1}{2}-\frac{1}{n}}
\qquad\mbox{as}\qquad \Gr\to\infty\,.
\ee
\end{theorem}
%%%%%%%%%%%%
\par\smallskip\noindent
It follows that
\bel{eq:gp07}\\
L \langle\kappa_n\rangle
%\leqslant L\langle\kappa_n^2\rangle^{1/2}
\leqslant c_{n}a_{\ell}^{\frac{3}{2}\left(1-\frac{1}{n}\right)} \Rey^{\frac{3}{4}-\frac{1}{2n}}\,[\ln(a_\ell^2\Rey)]^{\frac{1}{4}-\frac{1}{2n}}\,.
\ee
We want to compare this bound with a practical estimate for $L K_n$ under the assumption that $\ell\sim L/2\pi$. Consider the mean enstrophy dissipation rate
\begin{equation}
\eta_\nu=\nu L^{-2}\Langle H_2\Rangle\,.
\end{equation}
This is bounded \citep{ad06,gp07} and, at large $\Rey$, can be estimated as
\begin{equation}\label{eq:eta}
\eta_\nu \sim \frac{U^3}{\ell^3}\,.
\end{equation}
The flow is assumed to be isotropic and to have an
energy spectrum of the form
\bel{eq:spectrum}
\mathcal{E}(k) \sim \eta_\nu^{2/3}\ell^{3-q}k^{-q} \qquad (\ell^{-1}\leqslant k\leqslant k_{c})
\ee
with $1<q<5$.
The mean enstrophy dissipation rate can be obtained from the energy spectrum via the relation \citep{my75}\,:
\begin{equation}
\label{eq:eta-spectrum}
\eta_\nu = \nu \int_{\ell^{-1}}^{k_{c}} k^4 \mathcal{E}(k)dk \sim \nu \eta_\nu^{2/3}\ell^{3-q} k_{c}^{5-q}\,,
\end{equation}
where the contributions coming from wavenumbers $k > k_c$ have been ignored. Combining \eqref{eq:eta}, \eqref{eq:spectrum}, and \eqref{eq:eta-spectrum} yields
\begin{equation}
\label{eq:knu}
\ell k_{c} \sim \Rey^\frac{1}{5-q}\,.
\end{equation}
%For $q=3$, this estimate agrees with the usual prediction for the direct-cascade regime,
%i.e. $\ell k_{c}=\Rey^{2}$ \citep{km80,be12}.
By plugging \eqref{eq:spectrum} into \eqref{eq:K} and using \eqref{eq:knu}, we find\,:
\begin{equation}
\label{eq:scaling-Kn-2d}
L K_n \sim \ell K_n \sim (\ell k_{c})^{1-\frac{q-1}{2n}} \sim \Rey^{\frac{1}{5-q}-\frac{1}{2n}\left(\frac{q-1}{5-q}\right)}.
\end{equation}
We now compare \eqref{eq:scaling-Kn-2d} with \eqref{eq:gp07} and conclude that the Reynolds-number scaling of $K_n$ is consistent with that of $\langle\kappa_n\rangle$  provided that
\bel{eq:bound-2d}
q \leqslant \frac{11n-12}{3n-4}\,.
\ee
Since this must hold for all $n\geqslant 2$ and the right-hand side of \eqref{eq:bound-2d} is a decreasing function of $n$, 
we find that
\bel{qbd}
1< q \leqslant \frac{11}{3}\,.
\ee
\begin{remark}
The bound in \eqref{eq:bound-2d} can also be derived by comparing the high-$\Rey$ scaling of the $2(n-1)$-th moment of the enstrophy spectrum with the bound for $\langle\kappa_{n,1}\rangle^{2(n-1)}$ obtained in \citet{gp07}.
\end{remark}
\begin{remark}
The bound in \eqref{eq:bound-2d} agrees with a practical estimate of \citet{sf75} and a rigorous result of \citet{e96}. The exponent $-11/3$ also describes the energy spectrum of spiral structures in two-dimensional turbulence \citep{gilbert88}.
\end{remark}
In numerical simulations of isotropic turbulence, the following two types of forcing are commonly used\,: (i)
strictly monochromatic forcings  with a single wavenumber $\ell^{-1}$ and (ii)
forcings that maintain a constant energy injection rate $\epsilon$, i.e.
\bel{}
\bm f=\epsilon\,L^2\frac{\mathcal{P}\bm u}{\Vert\mathcal{P}\bm u\Vert_2}\,,
\ee
where the operator $\mathcal{P}$ projects the velocity field on a finite set of spatial modes.
For these forcings, it is possible to derive a more stringent bound on $q$.
Indeed, as $\Gr\to\infty$ the general estimate for
the mean enstrophy dissipation rate \citep{ad06,gp07}
\bel{}
\langle H_2\rangle \leqslant c\,\nu^2\ell^{-4}a_\ell^2\Rey^3
\ee
is replaced with \citep{ad06}
\bel{eq:monochromatic}
\langle H_2\rangle \leqslant c\,\nu^2\ell^{-4}a_\ell^2\Rey^2.
\ee
Using \eqref{eq:monochromatic} in the proof of Theorem~\ref{th:gp07} yields the following result.
\begin{theorem}
\label{th:gp07-monochromatic}
For $n\geqslant 2$ and a monochromatic or a constant-energy-input forcing
\bel{eq:kappa-monochromatic}
L^2\langle\kappa_n^2\rangle\leqslant c_{n}a_{\ell}^{3-3/n}\, \Rey \,[\ln(a_\ell^2\Rey)]^{1/2-1/n}
\quad\mbox{as}\quad\Gr\to\infty\,.
\ee
\end{theorem}
\bigskip
\noindent
By comparing \eqref{eq:kappa-monochromatic} with \eqref{eq:scaling-Kn-2d}, we find that for these types of forcing
\bel{}
q\leqslant 3+\frac{2}{n-1}\,,
\ee
which yields for every $n$
\bel{}
1<q\leqslant 3\,.
\ee 
Note that, up to logarithmic corrections, Kraichnan's prediction for the energy spectrum in the enstrophy-cascade range is
$\mathcal{E}(k)\sim k^{-3}$ \citep{k67,k71}. 

%%%%%%%%%%%
\subsubsection{\small\color{blue}The inverse cascade of energy}\label{2Dinverse}

To investigate the regime of the  inverse cascade, we study the behaviour of $\langle\kappa_n\rangle$ in the limit in which $a_\ell\to\infty$
while the Reynolds number based on the characteristic velocity at the forcing scale is $O(1)$. More precisely,
consider $u_f^2 = L^{-2}\langle\Vert\bm u_f\Vert_2^2\rangle$, where 
\begin{equation}
\bm u_f(\bm x,t)=\sum_{\vert\bm k\vert>\ell^{-1}} \mathrm{e}^{{\rm i} \bm k\cdot\bm x}\hat{\bm u}(\bm k,t)\,,
\end{equation}
and define
\begin{equation}
\Rey_f=\frac{u_f\ell}{\nu}\,.
\end{equation}
If $Re_f\sim 1$, then $\ell^{-1}\sim k_c$ and the direct cascade of enstrophy is negligible \citep[see][]{ab18}.
Furthermore, if the energy spectrum $\mathcal{E}(k)\sim k^{-q}$ for $L^{-1}\leqslant k\leqslant \ell^{-1}$ and is negligible otherwise,
it can be shown that \citep{sy94,t07}
\begin{equation}
U^2 = a_\ell^{q-1}u_f^2 
\end{equation}
and hence
\begin{equation}
\Rey = a_\ell^{(q-1)/2}\,\Rey_f\,.
\label{eq:Re-Ref}
\end{equation}
Therefore, in the regime considered here, $\Rey\sim a_\ell^{(q-1)/2}$.
We can now prove a bound on $\langle\kappa_n\rangle$ which is relevant to the energy cascading range.
\begin{theorem}
\label{th:kn-ic}
If $\mathcal{E}(k)$ is steeper than $k^{-5}$  and $\Rey_f=O(1)$, then for $n\geqslant 1$ and as $a_\ell\to\infty$
\bel{eq:kn-ic}
L^2\langle\kappa_{n}^2\rangle \leqslant c\, a_\ell^{n}\,\Rey\,.
\ee
\end{theorem}
\begin{proof}
Recall that
\begin{equation}
\langle F_2\rangle \leqslant c\,\langle H_2\rangle \leqslant  c\,\nu^2\ell^{-4}a_\ell^2 (\Rey^2+\Rey^3)\,,
\label{eq:F2}
\end{equation}
where the bound on $\langle H_2\rangle$ can be found in \citet{ad06} and Gibbon and Pavliotis \citeyearpar{gp07}.
Using the tighter bound for monochromatic or constant-energy-input forcings would not change the result in this case.
\par\smallskip%\noindent
In addition, the following form of the Brezis--Gallou\"{e}t inequality holds
\begin{equation}
\Vert\nabla\bm u\Vert_\infty \leqslant c\, F_2^{1/2}[1+\ln(L\kappa_{3,2})]^{1/2}\, .
\label{eq:BG}
\end{equation}
Now note that, as $a_\ell\to\infty$, the $F_n$ satisfy the same ladder as in \eqref{eq:ladder}  (see Appendix~\ref{appendix-A}). 
Thus, following \citet{gp07}, we divide through the ladder by $F_n$ and time average. We then use \eqref{eq:BG} together with $\kappa_{n,r}\leqslant\kappa_{n+1,n}$ for $2\leqslant r<n$ and Jensen's inequality on the logarithm to find
\begin{equation}
L^2\langle\kappa_{n,r}^2\rangle \leqslant
c\, L^2\nu^{-1} \langle F_2\rangle^{1/2} [1+\ln(L^2\langle\kappa_{n,r}^2\rangle)]^{1/2}+c\, a_\ell^{n+1}(1+\Rey).
\end{equation}
By using \eqref{eq:Re-Ref} and \eqref{eq:F2}, we can see that the first term on the right-hand side behaves as $a_\ell^{3+3(q-1)/4}\,\Rey_f^{3/2}$,
whereas the second behaves as $a_\ell^{n+1+(q-1)/2}\,\Rey_{f}$. Since $\Rey_f=O(1)$, $n\geqslant 3$, and $q<5$, the
second term dominates over the first. 
For $n>r\geqslant 2$ and in the limit $a_\ell\to\infty$ while $\Rey_f=O(1)$, we thus find
\bel{eq:n3}
L^2\langle\kappa_{n,r}^2\rangle \leqslant c\, a_\ell^{n+1}(1+\Rey).
\ee
By adapting the proofs of \citet{gp07} in the manner described in Appendix~\ref{appendix-A}, it is also possible to show that as $a_\ell\to\infty$
\bel{eq:F0+F1}
  \frac{1}{2}\dot{F}_0\leqslant -\nu F_{1} + \frac{F_0}{2\tau_0} 
\qquad
\text{and}
\qquad
  \frac{1}{2}\dot{F}_1\leqslant -\nu F_{2} + \frac{F_1}{2\tau_1}\,,
\ee
which imply
\bel{eq:n12}
L^2\langle \kappa_{1}^2\rangle \leqslant c\,a_\ell (1+\Rey) \qquad \text{and} \qquad
L^2\langle \kappa_{2,1}^2\rangle \leqslant c\,a_\ell^2 (1+\Rey)\,.
\ee
The advertised result follows from using \eqref{eq:n3} and \eqref{eq:n12} in
\begin{equation}
\langle \kappa_{n}^2\rangle 
=\left \langle \left(\frac{F_n}{F_1}\right)^{1/n}\left(\frac{F_1}{F_0}\right)^{1/n}  \right\rangle
\leqslant \langle \kappa_{n,1}^2\rangle^{(n-1)/n}\langle\kappa_{1}^2\rangle^{1/n}\,.
\end{equation}
\end{proof}
%%%%%%
\par\bigskip\noindent
We now move to the practical estimate for the moments of the spectrum. We remind the reader that  we are assuming that $\ell\ll L$ and $\ell\sim k_c$, so that the contribution from the spectrum at wavenumbers in the enstrophy cascading range is negligible. Assuming that $\mathcal{E}(k)\sim k^{-q}$ with $1<q<5$ in the range $L^{-1}\leqslant k\leqslant \ell^{-1}$, we find
\begin{equation}
L^{2n}K_n^{2n}=\dfrac{\displaystyle\int_{L^{-1}}^{\ell^{-1}} k^{2n}\mathcal{E}(k)dk}{\displaystyle\int_{L^{-1}}^{\ell^{-1}}\mathcal{E}(k)dk}
\sim \left(\frac{L}{\ell}\right)^{2n+1-q},
\end{equation}
or 
\begin{equation}
L K_n\sim a_\ell^{1-(q-1)/2n}.
\label{eq:Kn}
\end{equation}
In order to compare \eqref{eq:Kn} with the mathematical bound for $L\langle\kappa_{n}\rangle$, we recall \eqref{eq:Re-Ref}
and the assumption $\Rey_f=O(1)$.
Thus, \eqref{eq:Kn} can be recast as
\begin{equation}
L K_n\sim a_\ell^{1-(q-1)/2n-(q-1)/4}\Rey^{1/2}
\label{eq:Kn-recast}
\end{equation}
and the practical estimate for $L K_n$ is consistent with the bound for $L\langle\kappa_{n}\rangle$ if
\begin{equation}
q\geqslant \frac{2+5n-2n^2}{n+2}\,.
\end{equation}
Since the right-hand side is a decreasing function of $n$, this means that the constraint on $q$ is fixed by the $n=1$ case, i.e.
\begin{equation}
q\geqslant \frac{5}{3}.
\end{equation}
This lower bound agrees with an earlier result of \citet{t07} and with Kraichnan's prediction for the energy cascading range \citep{k67,k71}. 
That the bound is obtained for $n=1$ rather than considering the large-$n$ limit
is consistent with the fact that the inverse energy cascade is a large-scale phenomenon.

%%%%%%%%%%%%
\section{\large\color{blue}Burgers equation}\label{Burg}

All the quantities introduced in \S\ref{sect:NS} can be defined analogously for the Burgers equation by taking $d=1$ on the periodic interval $\mathcal{V}=[0,L]$
\bel{Burg1}
\partial_{t} u + u\,\partial_{x} u = \nu \partial_{x}^2 u + f\,.
\ee
In particular, we can again set $\delta=0$ in the definition of $\tau_n$.
\par\smallskip%\noindent
The following two lemmas can be proved by adapting the proofs for the 3$D$ Navier--Stokes equations \citep[see][]{dg02} to the Burgers equation\,:
%%%%%%%%%%%%%
\begin{lemma}
\label{lemma:ladder}
For $n\geqslant 1$ and as $\Gr\to\infty$, the $F_n$ satisfy the ladder in \eqref{eq:ladder}.
\end{lemma}
%%%%%%%%%%%%%
\begin{lemma}
\label{lemma:kappa1}
There exists a positive constant $c$ such that, as $\Gr\to\infty$,
\begin{equation}
\ell^2\left\langle\kappa_1^2\right\rangle\leqslant c\,a_\ell^{-1}\Rey\,.
\end{equation}
\end{lemma}
%%%%%%%%%%%%%
\noindent
We now prove the analogue of Theorems~\ref{th:dg02} and \ref{th:gp07} for the Burgers equation.
%%%%%%%%%%%%%
\begin{theorem}
\label{th:Burgers}
For $n\geqslant 2$, as $\Gr\to\infty$
\begin{equation}
\ell\Langle \kappa_{n}\right\rangle
\leqslant c_n\,a_{\ell}^{\frac{1}{3} - \frac{5}{6n}}\,\Rey^{1 - \frac{1}{2n}}\,.
\end{equation}
\end{theorem}
%%%%%%%%%%%%%
\begin{proof}
The inequality
\begin{equation}
\Vert \partial_x u\Vert_\infty \leqslant c\,\Vert\partial_x  u\Vert_2^{1/2}\Vert \partial^2_x u\Vert_2^{1/2}\leqslant c\,F_1^{1/4}F_2^{1/4}
\end{equation}
turns the ladder in \eqref{eq:ladder} into
\begin{equation}
\frac{1}{2}\dot{F}_n\leqslant -\nu F_{n+1}+c_n \left(F_1^{1/4}F_2^{1/4}+\tau_n^{-1}\right)F_n.
\end{equation}
By dividing through by $F_n$, time averaging, and noting that the forcing term is subdominant and can therefore be ignored, we find
\begin{equation}
\left\langle \kappa_{n+1,n}^2\right\rangle \leqslant 
\frac{c_n}{\nu}\Langle F_1^{1/4}F_2^{1/4}\Rangle = \frac{c_n}{\nu}\left\langle \left(\frac{F_2}{F_1}\right)^{1/4}
F_1^{1/2}\right\rangle = \frac{c_n}{\nu}\left\langle\kappa_{2,1}^{1/2} F_1^{1/2}\right\rangle \leqslant 
\frac{c_n}{\nu} \langle\kappa_{2,1}\rangle^{1/2} 
\langle F_1\rangle^{1/2}
\end{equation}
and hence, by using Jensen's inequality,
\begin{equation}
\label{eq:kn1n}
\left\langle \kappa_{n+1,n}^2\right\rangle \leqslant  c_n \nu^{-1} \langle\kappa_{2,1}^2\rangle^{1/4} \langle F_1\rangle^{1/2}.
\end{equation}
The estimate of \cite{df02} for the mean energy dissipation rate, and consequently the corresponding estimate for $\langle F_1\rangle$,
also hold for the Burgers equation.
As $\Gr\to\infty$, we thus have $\langle F_1\rangle\leqslant c\langle H_1\rangle\leqslant c\,\nu^2\ell^{-3}a_\ell \Rey^3$.
Inserting this estimate into \eqref{eq:kn1n} with $n=1$ yields as $\Gr\to\infty$
\begin{equation}
\label{eq:k21}
\Langle \kappa_{2,1}^2\Rangle \leqslant  c\,  \ell^{-2}a_{\ell}^{2/3} \Rey^{2}\,.
\end{equation}
Together \eqref{eq:kn1n} and \eqref{eq:k21} give
\begin{equation}
\label{eq:kn12}
\Langle \kappa_{n,1}^2\right\rangle\leqslant \Langle \kappa_{n+1,n}^2\right\rangle \leqslant c_n\,  \ell^{-2}a_{\ell}^{2/3} \Rey^2.
\end{equation}
We also have
\begin{equation}
\label{eq:k2n}
\Langle \kappa_{n}^2\right\rangle = \Langle \left(\frac{F_n}{F_0}\right)^\frac{1}{n}\Rangle =
\Langle \left(\frac{F_n}{F_1}\right)^\frac{1}{n}\left(\frac{F_1}{F_0}\right)^\frac{1}{n}\Rangle =
\Langle \kappa_{n,1}^\frac{2(n-1)}{n} \kappa_1^\frac{2}{n}\Rangle
\leqslant
\Langle \kappa_{n,1}^2\Rangle^\frac{n-1}{n} \Langle\kappa_1^2\Rangle^\frac{1}{n}.
\end{equation}
Therefore, by using \eqref{eq:kn12} and Lemma~\ref{lemma:kappa1}, we find
\begin{equation}
\Langle \kappa_{n}\right\rangle \leqslant \Langle \kappa_{n}^2\right\rangle^{1/2}\leqslant 
\Langle \kappa_{n,1}^2\Rangle^\frac{n-1}{2n} \Langle\kappa_1^2\Rangle^\frac{1}{2n}
\leqslant c\, \ell^{-1} a_{\ell}^{\frac{1}{3}-\frac{5}{6n}}\Rey^{1-\frac{1}{2n}}\,.
\end{equation}
\end{proof}%
We assume again that the flow is isotropic, the forcing is large-scale with $\ell\sim L/2\pi$, and the energy spectrum is as in \eqref{Eqdef} with $1<q<3$. By proceeding as for $d=3$ (see \S~\ref{sect:introduction}), we find
%\begin{equation}
%\label{eq:A-1D}
%A\sim\epsilon^{2/3}\ell^{5/3-q}.
%\end{equation}
%In addition
%\begin{equation}
%\label{eq:epsilon-spectrum}
%\epsilon = \nu \int_{\ell^{-1}}^{k_{c}} k^2 \mathcal{E}(k)dk \sim \nu A k_{c}^{3-q}
%\end{equation}
%and
%\begin{equation}
%\label{eq:epsilon-largeRe}
%\epsilon \sim \frac{U^3}{\ell}
%\end{equation}
%at large $\Rey$. 
%Combining \eqref{eq:A-1D}, \eqref{eq:epsilon-spectrum}, and \eqref{eq:epsilon-largeRe} yields
\begin{equation}
\label{eq:knu-1d}
\ell k_{c} \sim \Rey^\frac{1}{3-q}
\end{equation}
%For $q=3$, this estimate agrees with the usual prediction for the direct-cascade regime,
%i.e. $\ell k_{c}=\Rey^{2}$ \citep{km80,be12}.
%By plugging \eqref{eq:spectrum} into \eqref{eq:K} and using \eqref{eq:knu-1d}, we find\,:
and
\begin{equation}
\label{eq:scaling-Kn-1d}
\ell K_n \sim (\ell k_{c})^{1-\frac{q-1}{2n}} \sim \Rey^{\frac{1}{3-q}-\frac{1}{2n}\left(\frac{q-1}{3-q}\right)}.
\end{equation}
Therefore, after comparing \eqref{eq:scaling-Kn-1d}
with Theorem~\ref{th:Burgers}, we conclude that
the scaling of $\ell K_n$ is consistent with that of $\ell\langle\kappa_n\rangle$  provided that
\begin{equation}
\label{eq:bound-1d}
1<q\leqslant 2\,. 
\end{equation}
\begin{remark}
The energy spectrum of the Burgers equation for a large-scale forcing is known to behave as $k^{-2}$ \citep[e.g.][]{frsbp13,b14}.
\end{remark}

%%%%%%%%%%
\section{\large\color{blue}Shell models}\label{Shell}

In the `Sabra' shell model \citep{sabra}, the velocity variables $u_j$ are complex and satisfy the system of ordinary differential equations%
\footnote{The results would be the same for the Gledzer--Ohkitani--Yamada (GOY) shell model \citep{g73,yo87}.}
\begin{equation}
  \dot{u}^{\vphantom{*}}_{j} =
  i(a_1 k^{\vphantom{*}}_{j+1} u^*_{j+1} u^{\vphantom{*}}_{j+2}
  +a_2 k^{\vphantom{*}}_j u^{\vphantom{*}}_{j+1} u^*_{j-1}
  -a_3 k^{\vphantom{*}}_{j-1} u^{\vphantom{*}}_{j-1} u^{\vphantom{*}}_{j-2})
  -\nu k_j^2 u^{\vphantom{*}}_j + f^{\vphantom{*}}_{j}, \qquad j=1,2,3,\dots,
\label{eq:shell}
\end{equation}
where $u_j^*$ is the complex conjugate of $u_j$,
$\nu$ is the kinematic viscosity, $f_j$ are the forcing variables, and $k_j=k_0\lambda^j$ with $k_0>0$ and $\lambda>1$. The `boundary conditions' are $u_0=u_{-1}=0$, while the coefficients $a_1$, $a_2$, $a_3$ are real and such that $a_1 + a_2 + a_3 = 0$. This ensures 
that the kinetic energy 
\bel{KE1}
E=\Sum{j}\vert u_j\vert^{2}
\ee
is conserved when $\nu=0$ and $f_j=0$ for all $j$. Moreover, the time-averaged energy dissipation rate is
\bel{epsdef}
\epsilon = \nu\Langle\Sum{j}k_j^2\vert u^{\vphantom{*}}_j\vert^2\Rangle\,.
\ee
%\par\smallskip%\noindent
The forcing is assumed to be of the form $f_j=\mathcal{F}\phi_{j-j_f}$, where $\mathcal{F}$ is a complex constant
and $\phi_p=0$ for $p<0$ and $p>j_{\rm max}-j_f$. Therefore, $k_f=k_0\lambda^{j_f}$ 
and $k_{\rm max}=k_0\lambda^{j_{\rm max}}$ are the characteristic and maximum wavenumbers of the forcing, respectively.
Under these assumptions and if the initial energy is finite, the shell model has globally regular solutions \citep{clt06}.
Finally, $\Gr$ and $\Rey$ are defined 
as in \S\ref{sect:introduction} with $U = \sqrt{\langle E\rangle}$, $\ell=k^{-1}_{f}$, $a_\ell=k_f/k_1$, and $f=\vert\mathcal{F}\vert$.
\par\smallskip%\noindent
The shell-model analogues of $H_{n}$ and $F_n$ are
\bel{Fndefshell}
H_n=\Sum{j} k_j^{2n} \vert u_j^\phexp\vert^{2}\,,\qquad\quad
F_n = H_n + \tau_n^2 \Sum{j} k_j^{2n} \vert f_j^\phexp\vert^2, 
\ee
and $\tau_n$ is as in \eqref{eq:tau1} with $\delta=0$. As in the case of the Navier--Stokes equations, the definition of $\tau_n$ ensures that $\langle F_n\rangle$ and $\langle H_n\rangle$ scale in the same way as $\Gr\to\infty$. Indeed, $\epsilon$ satisfies an inequality analogous to \eqref{lower-eps} (see \eqref{eq:epsilon-shell} in Appendix~\ref{appendix-B}) which gives the same lower bound as in \eqref{lower-eps-2}. Using used a shell-model version of Poincar\'e's inequality
\bel{epsbd2}
\epsilon \leqslant \nu k_1^{-2(n-1)} \left\langle H_n\right\rangle\,.
\ee
we find
\beq{taudef}
\tau_n^2 \Sum{j} k_j^{2n} \vert f_j^\phexp\vert^2 &=& b_n\nu^{-2}k_f^{2n-4} a_\ell^{-2(n-1)}(1+\Rey)^{-2} f^2\nonumber\\
&=& b_{n}\nu^2 k_f^{2(n+1)} a_\ell^{-2(n-1)}(1+\Rey)^{-2}\Gr^{2} \\
&\leqslant& c\, b_n\nu^{-1} k_f^{2(n-1)} a_\ell^{-2(n-1)}\epsilon \leqslant  c\,b_n \langle H_n\rangle\,,\nonumber
\eeq
where $b_n=\sum_{p=0}^{j_{\rm max}-j_f} \lambda^{2np}\vert\phi_p\vert^2$. 
It was proved in \citet{vg21} that, analogously to the Navier--Stokes equations,
\bel{eq:Gr-Re-shell}
\Gr\leqslant c(\Rey+\Rey^2) 
\ee
and as $\Gr\to\infty$
\bel{eq:F1-shell}
\langle H_{1}\rangle \leqslant c\, \nu^2\ell^{-4}\Rey^3\,.
\ee
In addition, as $\Gr\to\infty$ the $F_n$ satisfy the same ladder of differential inequalities as in \eqref{eq:ladder} 
with $\Vert\nabla\bm u\Vert_\infty$ replaced with $\sup_{1\leqslant j\leqslant \infty} k_j|u_j|$\,:%
\footnote{\citet{vg21} proved the ladder for shell models with a single time scale $\tau$. The proof can be easily modified
to include an $n$-dependent time scale by following the same approach as in Appendix \ref{appendix-A}.}
\bel{eq:ladder-shell}
\frac{1}{2}\dot{F}_n\leqslant -\nu F_{n+1}+c_n\Big(\sup_{1\leqslant j\leqslant\infty}k_j\vert u_j\vert+\tau_n^{-1}\Big)F_n.
\ee
In shell models, the energy spectrum is defined as $\mathcal{E}(k_j)=k_j^{-1}\Langle\vert u_j(t)\vert^2\Rangle$ [Yamada and Ohkitani, \citeyear{yo87}]. Therefore,
in the limit $\Gr\to\infty$, the quantity $\kappa_{n}^{2n}=F_n/F_0$ behaves as the ratio of the $(2n+1)$-th to the first moment of the 
{\it instantaneous} energy spectrum. To obtain the  $\Rey$-scaling of $\Langle\kappa_n\Rangle$, we first need the shell-model analogue of 
Lemma~\ref{lemma:kappa1}.
%%%%%%%%%%%%%%%
\begin{lemma}
As $\Gr\to\infty$,
\begin{equation}
k_f^{-2}\langle\kappa_1^2\rangle \leqslant c\, a_\ell^{-1}\Rey.
\end{equation}
\label{lemma:shell-1}
\end{lemma}
%%%%%%%%%%%%%%%%%
\begin{proof}
The energy evolution equation for the shell model is
\bel{eq:energy}
\frac{d E}{d t}=-2\nu H_1 + \Sum{j}(f_j^\phexp u_j^*+f_j^* u_j^\phexp)\,.
\ee
We add and subtract $b_1\nu\tau_1^2k_f^{2}f^2$ to the right-hand side
and apply the Cauchy--Schwarz inequality to the forcing term to obtain
\begin{equation}
\frac{1}{2}\dot{F}_0 \leqslant -\nu F_1 + b_1\nu\tau_1^2k_f^{2}f^2 + b_0^{1/2} f H_0^{1/2}
\end{equation}
An application of  Young's inequality with parameter $g\tau_0^2$ yields
\bel{eq:F0-shell}
\frac{1}{2}\dot{F}_0 \leqslant -\nu F_1 + \frac{H_0}{2g\tau_0^2} 
+ \left(\frac{g}{2}+\frac{b_1\nu k_f^2}{b_0a_\ell^2}\right)b_0\tau_0^2 f^2,
\ee
where we have used $\tau_1=a_\ell^{-1}\tau_0$ and
$g$ is such that the coefficients of $H_0$ and $b_0\tau^2 f^2$ are the same\,:
\begin{equation}
\label{eq:g-shell}
g=-\frac{b_1\nu k_f^2}{b_0a_\ell^2} +\left(
\frac{b_1^2\nu^2 k_f^4}{b_0^2 a_\ell^4}+\frac{1}{\tau_0^2}\right)^{1/2}\,.
\end{equation}
Therefore, as $\Gr\to\infty$ we find $g\sim\tau_0^{-1}$  and \eqref{eq:F0-shell} becomes
\begin{equation}
\frac{1}{2}\dot{F}_0 \leqslant -\nu F_1 + c\,\tau_0^{-1}F_0\,.
\end{equation}
By dividing by $F_0$ and time averaging, we finally get
\begin{equation}
\langle \kappa_1^2\rangle \leqslant c\, \nu^{-1}\tau_0^{-1}\,. 
\end{equation}
The lemma is proved by replacing the definition of $\tau_0$.
\end{proof}
%%%%%%%%%%%%
\begin{lemma}For $n\geqslant 1$, as $\Gr\to\infty$
\label{lemma:shell-2}
\begin{equation}
k_f^{-1}\langle\kappa_{n,1}\rangle\leqslant c_n \Rey^{3/4} \,.
\end{equation}
\end{lemma}
%%%%%%%%%%%%
\begin{proof}
Dividing through \eqref{eq:ladder-shell} by $F_n$, time averaging, and ignoring the subdominant forcing term yields
\bel{shprf1}
\Langle\kappa_{n+1,n}^2\Rangle \leqslant c_n\nu^{-1} \left\langle\sup_{j\leqslant 1}k_j\vert u_j\vert \right\rangle
\leqslant c_n\nu^{-1}\Langle F_1^{1/2}\Rangle
\leqslant c_n\nu^{-1}\Langle F_1\Rangle^{1/2}.
\ee
The advertised result is obtained by using \eqref{eq:F1-shell}, $\langle F_1\rangle\leqslant c\,\langle H_1\rangle$, 
and $\Langle\kappa_{n,1}\Rangle \leqslant \langle\kappa_{n,1}^2\rangle^{1/2} \leqslant \langle\kappa_{n+1,n}^2\rangle^{1/2}$.
\end{proof}
\par\smallskip\noindent
The estimates in the above lemmas can be used to prove the following theorem\,:
%%%%%%%%%%%%%%%%
\begin{theorem}
For $n\geqslant 1$, as $\Gr\to\infty$
\bel{shthm1}
k_f^{-1}\langle\kappa_{n}\rangle\leqslant c_n a_\ell^{-1/2n}\Rey^{3/4-1/4n} \,.
\ee
\end{theorem}
%%%%%%%%%%%%%%
\begin{proof}
The proof is analogous to that of Theorem~1 in \citet{dg02}. To achieve this, first note that
\bel{shprf1-b}
\Langle \kappa_{n}\Rangle \leqslant \Langle \kappa_n^{\frac{2n}{2n-1}}\Rangle^{\frac{2n-1}{2n}} =
\Langle \kappa_{n,1}^\frac{2(n-1)}{2n-1} \left(\kappa_1^2\right)^\frac{1}{2n-1}\Rangle^{\frac{2n-1}{2n}}
\leqslant \Langle \kappa_{n,1}\Rangle^\frac{n-1}{n} \Langle\kappa_1^2\Rangle^\frac{1}{2n}\,,
\ee
and then use the estimates from Lemmas~\ref{lemma:shell-1} and ~\ref{lemma:shell-2}.
\end{proof}
\begin{remark}
The scaling of $\Langle\kappa_n\Rangle$ is the same as in Theorem~\ref{th:dg02-m} for $p=2$. This strengthens the parallel which was drawn in \citet{vg21} between shell models and the Navier--Stokes equations with suppressed velocity gradient fluctuations ($p=2$).
\end{remark}
By analogy with \eqref{eq:K}, we now define $K_n^{2n}$ as
\bel{shrm1}
K_n^{2n}=\frac{\Sum{j}k_j^{2n}\langle\vert u_j\vert^2\rangle}{\Sum{j}\langle\vert u_j\vert^2\rangle}=
\frac{\Sum{j}k_j^{2n+1}\mathcal{E}(k_j)}{\Sum{j}k_j\mathcal{E}(k_j)}\,.
\ee
We also assume $k_f=k_1$ and that there exists $k_{c}=k_0\lambda^{j_c}$ such that $\mathcal{E}(k_j)$
decays rapidly for $k_j>k_c$, while
\bel{shrm2}
\mathcal{E}(k_j) \sim
%\begin{cases}
Ak_j^{-q}, \qquad 1\leqslant j\leqslant j_c
%\\
%0, &  j \geqslant j_c
%\end{cases}
\ee
with $1<q<3$ and $A\sim\epsilon^{2/3}k_f^{q-5/3}\sim U^2 k_f^{q+1/3}$ at large $\Rey$. 
In addition
\bel{shrm3}
\epsilon = \nu \Sum{j}k_j^3 \mathcal{E}(k_j) \sim \nu A k_{c}^{3-q},
\ee
whence $k_{c}/k_f\sim \Rey^\frac{1}{3-q}$. We thus find
\bel{shrm4}
K_n^{2n}\sim k_{c}^{2n+1-q}k_f^{q-1}
\ee
and hence
\bel{shrm5}
K_n/k_f\sim \Rey^{\frac{1}{3-q}-\frac{q-1}{2n(3-q)}}.
\ee
We now follow the approach used for the Navier--Stokes and Burgers equations and compare the $\Rey$-scaling of $K_n$ and $\Langle\kappa_n\Rangle$. 
%even though in principle this does not correspond to a rigorous manipulation of the time averages. 
The two scalings are consistent if
\bel{shrm6}
1<q\leqslant \frac{5}{3}\,.
\ee
\begin{remark}
In the turbulent regime, the GOY and Sabra shell models display a $k^{-5/3}$ inertial-range spectrum \citep{yo87,sabra}. Moreover, in the inviscid unforced case, they possess fixed-point solutions with an energy spectrum scaling as $k^{-5/3}$\citep{bjpv98}.
\end{remark}

%%%%%%%%%%%%%%%%
\section{\large\color{blue}Summary and conclusion}

This paper has developed the method of \citet{dg02} in which a sequence of time-dependent wavenumbers, or inverse length scales $\kappa_n(t)$, was originally used to extract a spectrum from the $3D$ Navier--Stokes equations. These wavenumbers are ratios of volume integrals of velocity derivatives. For the $3D$ Navier--Stokes equations, and a version of them where large fluctuations of the velocity gradient are suppressed, they obtained rigorous bounds for the time average of $\kappa_n(t)$ in terms of $\Rey$. They interpreted the wavenumbers $\kappa_n(t)$ as moments of the energy spectrum and the bounds on the time average of these were then used to infer the slope of the energy spectrum in the inertial range of a turbulent velocity field. Since the $\kappa_n(t)$ are based on Navier--Stokes weak solutions, this approach connects empirical predictions of the energy spectrum with the mathematical analysis of the Navier--Stokes equations.
\par\smallskip%\noindent
We have extended these methods to other hydrodynamic equations that display a turbulent regime at high $\Rey$, namely the $2D$ Navier--Stokes equations, the Burgers equation, and shell models. The results are summarized in Table~\ref{table}. Previous predictions for the energy spectrum are recovered within the same mathematical framework, which confirms the appropriateness of $\Langle \kappa_n\Rangle$ as
quantities suitable for the rigorous study of the energy spectrum of hydrodynamic partial differential equations.

%%%%
\section*{\large\color{blue}Acknowledgments}

The authors are grateful to Samriddhi Sankar Ray and Pierre-Louis Sulem for useful discussions.
This research was supported in part by the CNRS--Imperial Collaboration Fund and
the International Centre for Theoretical Sciences (ICTS) for the online program 
`Turbulence: Problems at the Interface of Mathematics and Physics' (ICTS/TPIMP2020/12).

%%%%
\appendix
\section{\large\color{blue}Proof of the $F_n$ ladder}\label{appendix-A}

\renewcommand{\theequation}{\thesection.\arabic{equation}}
\setcounter{equation}{0}

Consider the ladder of inequalities for $n\geqslant 1$ \citep{dg95}\,:
\bel{laddapp1}
\frac{1}{2}\dot{H}_n\leqslant -\nu H_{n+1} + c_n \Vert\nabla\bm u\Vert_\infty H_n+ H_n^{1/2}\Vert\nabla^n\bm f\Vert_2\,.
\ee
In the case $d=2$ the time differentiation of the higher order $H_{n}$ is legal because the Navier-Stokes equations are regular. In the case $d=3$ the result is formally true if one assumes there is a solution with sufficiently long interval of regularity. We proceed on this basis noting, however, that the estimates for the time-averages achieved in this paper can be shown to be true for weak solutions \citep{JDG2019}. In turn, these are based on the work of \citet{FGT}. 
\par\medskip%\noindent
Add and subtract $\nu\tau_{n+1}^2\Vert\nabla^{n+1}\bm f\Vert_2^2$ to obtain
\begin{equation}
\frac{1}{2}\dot{F}_n\leqslant -\nu F_{n+1} + c_n \Vert\nabla\bm u\Vert_\infty F_n+ H_n^{1/2}\Vert\nabla^n\bm f\Vert_2
+\nu\tau_{n+1}^2\Vert\nabla^{n+1}\bm f\Vert_2^2\,.
\end{equation}
Now apply Young's inequality with parameter $g\tau_n^2$ to the last two terms of the right-hand side and use $\tau_{n+1}=a_\ell^{-1}\tau_n$ and
$\Vert\nabla^{n+1}\bm f\Vert_2= \ell^{-2}\Vert\nabla^{n}\bm f\Vert_2$\,:
\begin{align}
H_n^{1/2}\Vert\nabla^n\bm f\Vert_2 +\nu\tau_{n+1}^2\Vert\nabla^{n+1}\bm f\Vert_2^2
&\leqslant
\frac{1}{2g\tau_n^2} H_n+\frac{g\tau_n^2}{2}\Vert\nabla^n\bm f\Vert_2 
+\nu\tau_{n+1}^2\Vert\nabla^{n+1}\bm f\Vert_2^2
\\
&= \frac{1}{2g\tau_n^2} H_n+\tau_n^2\left(\frac{g}{2}+\frac{\nu}{a_\ell^2\ell^2}\right)\Vert\nabla^{n}\bm f\Vert_2^2
\end{align}
In order to have the same coefficients for $H_n$ and $\tau_n^2\Vert\nabla^{n}\bm f\Vert_2^2$ and thus form $F_n$, we must take
\begin{equation}
g= - \frac{\nu}{a_\ell^2\ell^2}+\left(\frac{\nu^2}{a_\ell^4\ell^4}+\frac{1}{\tau_n^2}\right)^{1/2}\,.
%\tau_n^{-1}
%\left(
%\sqrt{1+\frac{1}{a_\ell^{2n+3}\Gr^{1+2\delta}}}
%-\frac{1}{a_\ell^{n+3/2}\Gr^{1/2+\delta}}
%\right).
\end{equation}
As $\Gr\to\infty$ or $a_\ell\to\infty$, we find $g\sim\tau_n^{-1}$ for all $n\geqslant 1$.

%%%
\section{\large\color{blue}Proof of an inequality for $\epsilon$ in shell models}
\label{appendix-B}

\renewcommand{\theequation}{\thesection.\arabic{equation}}
\setcounter{equation}{0}

The proof of the analogue of \eqref{lower-eps} for shell models follows the strategy used by \citet{df02} for the $3D$ Navier--Stokes equations.
First define the constants
\beq{BCDdef}
  B_\lambda&=&[(\vert a_1\vert+\vert a_2\vert)\lambda^{-1}+\vert a_1+a_2\vert],\nonumber\\
  C_{M}&=&\sum_{m=0}^\infty \lambda^{2mM}\vert\phi_m\vert^2,\\
  D_{M} &=& \sup_{m\geqslant 0}\lambda^{-m(2M-1)}\vert\phi_m\vert\,,\nonumber
\eeq
where $M$ is any real number such that $C_M$ and $D_M$ are bounded.
In particular, the following equality \citep{vg21} will be useful later\,:
\begin{equation}
      \Sum{j}k_j^{2M}\vert f_j\vert^2  = C_{M} f^2 k_f^{2M}\,.
      \label{eq:forcing-shell}
\end{equation}
Now multiply Eq.~\eqref{eq:shell} by $k_j^{-2M} f_j^*$, sum over $j$, and average over
time:
\begin{multline}\label{eq:equality-F}
%  \left\langle
  \Sum{j}k_j^{-2M}\vert f_j\vert^2
%  \right\rangle
=\left\langle\nu\Sum{j}k_j^{2-2M}u_j f_j^*\right\rangle \\
-\left\langle i \Sum{j}k_j^{-2M}f_j^*(a_1k_{j+1}u^*_{j+1}u_{j+2}+a_2 k_j u_{j+1}u^*_{j-1}
  -a_3k_{j-1}u_{j-1}u_{j-2})\right\rangle\,.
\end{multline}
Rearranging the terms in the first time average on the right-hand side and using the Cauchy--Schwartz inequality and \eqref{eq:forcing-shell} yields
\beq{eq:viscous}
    \left\vert
    \left\langle \nu\Sum{j}k_j^{2-2M}u_j f_j^*\right\rangle\right\vert &=&\nu f\left\vert
    \left\langle\Sum{j}(k_j u_j)(k_j^{1-2M}\phi_{j-j_f}^*)\right\rangle\right\vert\nonumber\\
    &\leqslant& \nu^{1/2}\epsilon^{1/2}\sqrt{C_{1-2M}}\, f k_f^{1-2M}.
\eeq
The second time average can again be estimated by using the Cauchy--Schwartz inequality.
Consider for instance the term with coefficient $a_1$\,:
\bel{a1}
\begin{split}
\left\vert\left\langle
i a_1\Sum{j}k_j^{-2M}f_j^*k_{j+1}u^*_{j+1}u_{j+2}
\right\rangle\right\vert
&=\frac{\vert a_1\vert f}{\lambda}
\left\vert\left\langle
\Sum{j}u^*_{j+1}(k_{j+2}u_{j+2})(k_j^{-2M}\phi_{j-j_f}^*)
\right\rangle\right\vert
\\[2mm]
&\leqslant \frac{\vert a_1\vert}{\lambda}
\left(\frac{\epsilon}{\nu}\right)^{1/2}
D_{M+\frac{1}{2}} Uf  k_f^{-2M}.
\end{split}
\ee
Likewise we have
\begin{equation}
\left\vert\left\langle
i a_2\Sum{j}k_j^{-2M}f_j^* k_{j}u_{j+1}u^*_{j-1}
\right\rangle\right\vert
\leqslant \frac{\vert a_2\vert}{\lambda}
\left(\frac{\epsilon}{\nu}\right)^{1/2}
D_{M+\frac{1}{2}} Uf  k_f^{-2M}
\end{equation}
and
\begin{equation}\label{eq:nonlinear}
\left\vert\left\langle
i a_3\Sum{j}k_j^{-2M}f_j^* k_{j-1}u_{j-1}u_{j-2}
\right\rangle\right\vert
\leqslant\vert a_1+a_2\vert
\left(\frac{\epsilon}{\nu}\right)^{1/2}
D_{M+\frac{1}{2}} Uf  k_f^{-2M}.
\end{equation}
By combining \eqref{eq:forcing-shell} and the bounds in \eqref{eq:viscous} to
\eqref{eq:nonlinear}, we find
\begin{equation}
  C_{-M}f \leqslant C_{1-2M}^{1/2}\,
  \nu^{1/2}\epsilon^{1/2} k_f
  +B_\lambda D_{M+\frac{1}{2}}
  \left(\frac{\epsilon}{\nu}\right)^{1/2} U\,.
  \label{eq:epsilon-shell}
\end{equation}

%%%%
\bibliographystyle{plainnat} %{unsrtnat}
%\setcitestyle{authoryear,open={(},close={)}}
\bibliography{bdspec}

\begin{thebibliography}{36}
\providecommand{\natexlab}[1]{#1}
\providecommand{\url}[1]{\texttt{#1}}
\expandafter\ifx\csname urlstyle\endcsname\relax
  \providecommand{\doi}[1]{doi: #1}\else
  \providecommand{\doi}{doi: \begingroup \urlstyle{rm}\Url}\fi

\bibitem[Alexakis and Biferale(2018)]{ab18}
A.~Alexakis and L.~Biferale.
\newblock Cascades and transitions in turbulent flows.
\newblock \emph{Phys. Rep.}, 767--769:\penalty0 1--102, 2018.

\bibitem[Alexakis and Doering(2006)]{ad06}
A.~Alexakis and C.~R. Doering.
\newblock Energy and enstrophy dissipation in steady state 2d turbulence.
\newblock \emph{Phys. Lett. {\rm A}}, 359:\penalty0 652--657, 2006.

\bibitem[Bardos and Titi(2013)]{bt13}
C.~W. Bardos and E.~S. Titi.
\newblock Mathematics and turbulence: where do we stand?
\newblock \emph{J. Turbul.}, 14:\penalty0 42--76, 2013.

\bibitem[Boffetta and Ecke(2012)]{be12}
G.~Boffetta and R.~E. Ecke.
\newblock Two-dimensional turbulence.
\newblock \emph{Annu. Rev. Fluid Mech.}, 44:\penalty0 427--451, 2012.

\bibitem[Bohr et~al.(1998)Bohr, Jensen, Paladin, and Vulpiani]{bjpv98}
T.~Bohr, M.~H. Jensen, G.~Paladin, and A.~Vulpiani.
\newblock \emph{Dynamical Systems Approach to Turbulence}.
\newblock Cambridge University Press, Cambridge, 1998.

\bibitem[Boritchev(2014)]{b14}
A.~Boritchev.
\newblock Turbulence for the generalised {B}urgers equation.
\newblock \emph{Russ. Math. Surv.}, 69:\penalty0 957--994, 2014.

\bibitem[Constantin(2006)]{constantin2006}
P.~Constantin.
\newblock {Euler Equations, Navier-Stokes Equations and Turbulence}.
\newblock In M.~Cannone and T.~Miyakawa, editors, \emph{Mathematical Foundation
  of Turbulent Viscous Flows}, volume 1871 of \emph{Lecture Notes in
  Mathematics}, pages 1--43. Springer, Berlin Heidelberg, 2006.

\bibitem[Constantin et~al.(2006)Constantin, Levant, and Titi]{clt06}
P.~Constantin, B.~Levant, and E.~S. Titi.
\newblock Analytic study of shell models of turbulence.
\newblock \emph{Physica~{D}}, 219:\penalty0 120--141, 2006.

\bibitem[Doering(2009)]{doering2009}
C.~R. Doering.
\newblock The 3{D} {N}avier--{S}tokes problem.
\newblock \emph{Annu. Rev. Fluid Mech.}, 41:\penalty0 109--128, 2009.

\bibitem[Doering and Foias(2002)]{df02}
C.~R. Doering and C.~Foias.
\newblock Energy dissipation in body-forced turbulence.
\newblock \emph{J. Fluid Mech.}, 467:\penalty0 289--306, 2002.

\bibitem[Doering and Gibbon(1995)]{dg95}
C.~R. Doering and J.~D. Gibbon.
\newblock \emph{Applied Analysis of the {N}avier--{S}tokes Equations}.
\newblock Cambridge University Press, Cambridge, 1995.

\bibitem[Doering and Gibbon(2002)]{dg02}
C.~R. Doering and J.~D. Gibbon.
\newblock Bounds on moments of the energy spectrum for weak solutions of the
  three-dimensional {N}avier--{S}tokes equations.
\newblock \emph{Physica {\rm D}}, 165:\penalty0 163--175, 2002.

\bibitem[Eyink(1996)]{e96}
G.~L. Eyink.
\newblock Exact results on stationary turbulence in 2{D}: consequences of
  vorticity conservation.
\newblock \emph{Physica {\rm D}}, 91:\penalty0 97--142, 1996.

\bibitem[Foias et~al.(1981)Foias, Guillop\'e, and Temam]{FGT}
C.~Foias, C.~Guillop\'e, and R.~Temam.
\newblock New a priori estimates for the navier-stokes equations in dimension
  3.
\newblock \emph{Comm. PDE}, 6:\penalty0 329--359, 1981.

\bibitem[Foias et~al.(2001)Foias, Manley, Rosa, and Temam]{foias2001}
C.~Foias, O.~Manley, R.~Rosa, and R.~Temam.
\newblock \emph{Navier--Stokes Equations and Turbulence}.
\newblock Cambridge University Press, Cambridge, 2001.

\bibitem[Frisch(1995)]{frisch1995}
U.~Frisch.
\newblock \emph{Turbulence: The Legacy of A.~N.~Kolmogorov}.
\newblock Cambridge University Press, Cambridge, 1995.

\bibitem[Frisch et~al.(2013)Frisch, Ray, Sahoo, Banerjee, and Pandit]{frsbp13}
U.~Frisch, S.~S. Ray, G.~Sahoo, D.~Banerjee, and R.~Pandit.
\newblock Real-space manifestations of bottlenecks in turbulence spectra.
\newblock \emph{Phys. Rev. Lett.}, 110:\penalty0 064501, 2013.

\bibitem[Gibbon(2019)]{JDG2019}
J.~D. Gibbon.
\newblock Weak and strong solutions of the $3d$ navier-stokes equations and
  their relation to a chessboard of convergent inverse length scales.
\newblock \emph{J. Nonlin. Sci.}, 29:\penalty0 215--228, 2019.

\bibitem[Gibbon and Pavliotis(2007)]{gp07}
J.~D. Gibbon and G.~A. Pavliotis.
\newblock Estimates for the two-dimensional {N}avier--{S}tokes equations in
  terms of the {R}eynolds number.
\newblock \emph{J. Math. Phys.}, 48:\penalty0 065202, 2007.

\bibitem[Gibbon et~al.(2016)Gibbon, A, Krstulovic, Pandit, Politano, Ponty,
  Pouquet, Sahoo, and Stawarz]{ggkppppss16}
J.~D. Gibbon, Gupta A, G.~Krstulovic, R.~Pandit, H.~Politano, Y.~Ponty,
  A.~Pouquet, G.~Sahoo, and J.~Stawarz.
\newblock Depletion of nonlinearity in magnetohydrodynamic turbulence: Insights
  from analysis and simulations.
\newblock \emph{Phys. Rev. E}, 93:\penalty0 043104, 2016.

\bibitem[Gilbert(1988)]{gilbert88}
A.~D. Gilbert.
\newblock Spiral structures and spectra in two-dimensional turbulence.
\newblock \emph{J. Fluid Mech.}, 193:\penalty0 475--497, 1988.

\bibitem[Gledzer(1973)]{g73}
E.~B. Gledzer.
\newblock System of hydrodynamic type admitting two quadratic integrals of
  motion.
\newblock \emph{Sov. Phys. Dokl.}, 18:\penalty0 216--217, 1973.

\bibitem[Kellay and Goldburg(2002)]{kg02}
H.~Kellay and W.~I. Goldburg.
\newblock Two-dimensional turbulence: a review of some recent experiments.
\newblock \emph{Rep. Progr. Phys.}, 65:\penalty0 845--894, 2002.

\bibitem[Kraichnan(1967)]{k67}
R.~H. Kraichnan.
\newblock Inertial ranges in two‐dimensional turbulence.
\newblock \emph{Phys. Fluids}, 10:\penalty0 1417--1423, 1967.

\bibitem[Kraichnan(1971)]{k71}
R.~H. Kraichnan.
\newblock Inertial-range transfer in two- and three-dimensional turbulence.
\newblock \emph{J. Fluid Mech.}, 47:\penalty0 525--535, 1971.

\bibitem[Kraichnan and Montgomery(1980)]{km80}
R.~H. Kraichnan and D.~Montgomery.
\newblock Two-dimensional turbulence.
\newblock \emph{Rep. Progr. Phys.}, 43:\penalty0 547--619, 1980.

\bibitem[Kuksin and Shirikyan(2012)]{ks12}
S.~Kuksin and A.~Shirikyan.
\newblock \emph{Mathematics of Two-Dimensional Turbulence}.
\newblock Cambridge University Press, Cambridge, 2012.

\bibitem[L'vov et~al.(1998)L'vov, Podivilov, Pomyalov, Procaccia, and
  Vandembroucq]{sabra}
V.~S. L'vov, E.~Podivilov, A.~Pomyalov, I.~Procaccia, and D.~Vandembroucq.
\newblock Improved shell model of turbulence.
\newblock \emph{Phys. Rev. E}, 58:\penalty0 1811--1822, 1998.

\bibitem[Monin and Yaglom(1975)]{my75}
A.~S. Monin and A.~M. Yaglom.
\newblock \emph{Statistical Fluid Mechanics}, volume~2.
\newblock MIT Press, Cambridge, MA, 1975.

\bibitem[Smith and Yakhot(1994)]{sy94}
L.~M. Smith and V.~Yakhot.
\newblock Finite-size effects in forced two-dimensional turbulence.
\newblock \emph{J. Fluid Mech.}, 274:\penalty0 115--138, 1994.

\bibitem[Sulem and Frisch(1975)]{sf75}
P.-L. Sulem and U.~Frisch.
\newblock Bounds on energy flux for finite energy turbulence.
\newblock \emph{J. Fluid Mech.}, 72:\penalty0 417--423, 1975.

\bibitem[Tabeling(2002)]{t02}
P.~Tabeling.
\newblock Two-dimensional turbulence: a physicist approach.
\newblock \emph{Phys. Rep.}, 362:\penalty0 1--62, 2002.

\bibitem[Tran(2007)]{t07}
C.~V. Tran.
\newblock Constraints on inertial range scaling laws in forced two-dimensional
  {Navier--Stokes} turbulence.
\newblock \emph{Phys. Fluids}, 19:\penalty0 108109, 2007.

\bibitem[Verma(2019)]{verma2019}
M.~Verma.
\newblock \emph{Energy Transfers in Fluid Flows}.
\newblock Cambridge University Press, Cambridge, 2019.

\bibitem[Vincenzi and Gibbon(2021)]{vg21}
D.~Vincenzi and J.~D. Gibbon.
\newblock How close are shell models to the 3{D} {N}avier--{S}tokes equations?
\newblock \emph{Nonlinearity}, 34:\penalty0 5821--5843, 2021.

\bibitem[Yamada and Ohkitani(1987)]{yo87}
M.~Yamada and K.~Ohkitani.
\newblock Lyapunov spectrum of a chaotic model of three-dimensional turbulence.
\newblock \emph{J. Phys. Soc. Japan}, 56:\penalty0 4210--4213, 1987.

\end{thebibliography}

\end{document}